\documentclass[letterpaper, 10 pt, conference]{ieeeconf}  
\IEEEoverridecommandlockouts                              
\title{\LARGE \bf
Safe and Efficient Model Predictive Control Using Neural Networks: An Interior Point Approach
}

\author{Daniel Tabas and Baosen Zhang
\thanks{This work is partially supported by the National Science Foundation Graduate Research Fellowship Program under Grant No. DGE-1762114, NSF grants ECCS-1930605 and ECCS-2023531. Any opinions, findings, conclusions, or recommendations expressed in this material are those of the authors and do not necessarily reflect the views of the National Science Foundation.}
\thanks{Authors are with the department of Electrical and Computer Engineering, University of Washington, Seattle, WA, United States. $\{$dtabas, zhangbao$\}$@uw.edu.
}}

\usepackage{amsmath,amssymb,amsfonts,graphicx,comment}
\usepackage{xcolor}
\usepackage{cite}

\begin{document}

\newcommand{\N}{\mathbb{N}}
\newcommand{\Z}{\mathbb{Z}}
\newcommand{\R}{\mathbb{R}}
\newcommand{\any}{\text{ $\forall$ }}
\newcommand{\e}{\text{e}}
\newcommand{\E}{\mathcal{E}}
\newcommand\m[1]{\begin{bmatrix}#1\end{bmatrix}}

\newcommand{\X}{\mathcal{X}}
\newcommand{\U}{\mathcal{U}}
\newcommand{\D}{\mathcal{D}}
\renewcommand{\int}{\textbf{int }}
\newcommand{\B}{\mathbb{B}}
\newcommand{\F}{\mathcal{F}}
\newcommand{\T}{\mathcal{T}}
\newcommand{\V}{\mathcal{V}}
\renewcommand{\P}{\mathcal{P}}
\newcommand{\Q}{\mathcal{Q}}
\newcommand{\Y}{\mathcal{Y}}
\renewcommand{\S}{\mathcal{S}}
\newcommand{\M}{\mathcal{M}}

\renewcommand{\u}{\textbf{u}}
\newcommand{\x}{\textbf{x}}

\newcommand{\todo}[1]{\textcolor{red}{#1}}

\newtheorem{proposition}{Proposition}
\newtheorem{definition}{Definition}

\newcommand{\revision}[1]{\textcolor{black}{#1}}

\newcommand{\newrevision}[1]{\textcolor{black}{#1}}

\maketitle
\thispagestyle{empty}
\pagestyle{empty}

\begin{abstract}

Model predictive control (MPC) provides a useful means for controlling systems with constraints, but suffers from the computational burden of repeatedly solving an optimization problem in real time. Offline (explicit) solutions for MPC attempt to alleviate real time computational challenges using either multiparametric programming or machine learning. The multiparametric approaches are typically applied to linear or quadratic MPC problems, while learning-based approaches can be more flexible and are less memory-intensive. Existing learning-based approaches offer significant speedups, but the challenge becomes ensuring constraint satisfaction while maintaining good performance. In this paper, we provide a neural network parameterization of MPC policies that explicitly encodes the constraints of the problem. By exploring the interior of the MPC feasible set in an unsupervised learning paradigm, the neural network finds better policies faster than projection-based methods and exhibits substantially shorter solve times. We use the proposed policy to solve a robust MPC problem, and demonstrate the performance and computational gains on a standard test system. 

\end{abstract}


\section{Introduction}

Model predictive control (MPC) \cite{Rawlings2019} is a powerful technique for controlling systems that are subject to state and input constraints, such as agricultural \cite{Ding2018}, automotive \cite{Hrovat2012}, 
and energy systems \cite{Ademola-Idowu2021}. However, many applications require fast decision-making which may preclude the possibility of repeatedly solving an optimization problem online \cite{Alessio2009}.

A popular approach for accelerating MPC is to move as much computation offline as possible~\cite{Alessio2009,Zeilinger2011}. These techniques, known as explicit MPC, involve precomputing the solution to the MPC problem over a range of parameters or initial conditions. Most of the research effort has focused on problems with linear dynamics and constraints, and linear or quadratic cost functions. In this case,  the explicit MPC solution is a piecewise affine (PWA) function defined over a polyhedral partition of the state constraints. However, many of the applications of interest have cost functions that are not necessarily linear or quadratic, or even convex. Further, the memory required to store the partition and affine functions can be prohibitive even for modestly-sized problems. 


In order to reduce the complexity of explicit MPC, the optimal offline solution can be approximated. Approximations generally fall into two categories: partition-based solutions  \cite{Jones2007,Johansen2004,Grancharova2009} that generate piecewise control laws over coarser state space partitions, and learning-based solutions \cite{Akesson2006,Chen2018d,Parisini1995,Domahidi2011} that use function approximation to compactly represent the optimal MPC policy. In this paper, we focus on the latter with the key contribution of ensuring constraint satisfaction while exploring all feasible policies.

Constraint satisfaction is crucial in many engineering applications, and the ability of MPC to enforce constraints is a major factor in its popularity. However, it is not straightforward to guarantee that a learning-based solution will satisfy constraints. The main challenge arises from the fact that while neural networks can limit their outputs to be in simple regions, there is no obvious way of forcing complex constraint satisfaction at the output. In~\cite{Akesson2006,Maddalena2020}, supervised and unsupervised learning were used to approximate the solution of MPCs, but did not provide any feasibility guarantees. By contrast, \cite{Chen2018d} trains an NN using a policy gradient approach, and guarantees feasibility by projecting the NN output into the feasible action set. However, this extra optimization step slows down the speed of online implementation, making it difficult to use in applications that require high-frequency solutions~\cite{Zheng2020}. Supervised learning approaches that provide safety guarantees \cite{Domahidi2011, Parisini1995} rely on a choice of MPC oracle that is not obvious when persistent disturbances are present.

In this paper, we propose an NN architecture for approximating explicit solutions to finite-horizon MPC problems with linear dynamics, linear constraints, and arbitrary differentiable cost functions. The proposed architecture guarantees constraint satisfaction without relying on projections or MPC oracles. By exploring the \emph{interior} of the feasible set, we demonstrate faster training and evaluation, and comparable closed-loop performance relative to other NN architectures. 

The proposed approach has parallels in interior point methods for convex optimization \cite{Boyd2009}. Interior point methods first solve a \emph{Phase I} problem to find a strictly feasible starting point. This solution is used to initialize the \emph{Phase II} algorithm for optimizing the original problem. Our approach accelerates both phases. The Phase I solution is given by a simple function (e.g., affine map) and the Phase II problem is solved using an NN architecture that can encode arbitrary polytopic constraints (Fig. \ref{fig:intro}).

The Phase II solution builds on a technique first proposed in \cite{Tabas2021a}, which uses a \emph{gauge map} to establish equivalence between compact, convex sets. \revision{With respect to \cite{Tabas2021a}, the current work has three novel aspects. First, the reinforcement learning (RL) algorithm in \cite{Tabas2021a} only uses information about the constraints, and does not use information about the cost function or dynamics. The resulting policy is safe, but can exhibit suboptimal performance. The MPC formulation in the current paper gives rise to a training algorithm that can exploit knowledge about the system, improving performance. 
Second, the MPC formulation permits explicit consideration for future time steps. The RL formulation cannot optimize entire trajectories due to the presence of constraints. This inability to ``look ahead'' again limits the performance of the RL algorithm. 
Finally, the previous work required a Phase I that used a linear feedback to find a strictly feasible point. A linear feedback, however, may not exist for some problems. The current work proposes a more general class of Phase I solutions (piecewise affine), while providing a way to manage the complexity of the Phase I solution.}


\begin{figure}[b]
    \centering
    \includegraphics[height=3cm]{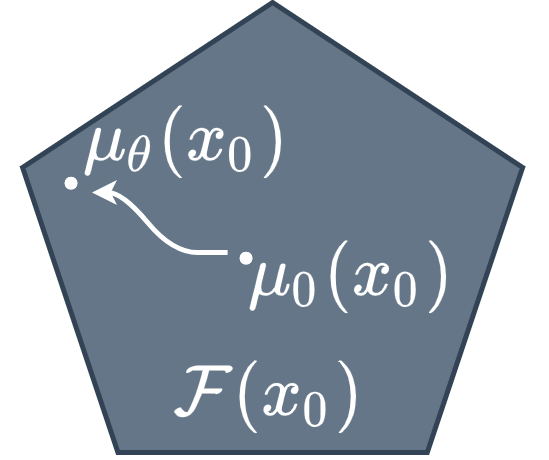}
    \caption{Illustration of the interior point approach to learning-based MPC. The set $\F(x_0)$ represents the MPC feasible set, while $\mu_0(x_0)$ and $\mu_\theta(x_0)$ are control input sequences representing solutions to the Phase I and Phase II problems, respectively. The neural network $\mu_\theta$ moves the Phase I solution to a more optimal solution.}
    \label{fig:intro}
\end{figure}

We demonstrate the effectiveness of the proposed technique on a 3-state test system, and compare to standard projection- and penalty-based approaches for learning with constraints. The results show that the proposed technique achieves Pareto efficiency in terms of closed-loop performance and online computation effort. All code is available at \texttt{github.com/dtabas/gauge\_networks}.


\textit{Notation:}
The \revision{$p$}-norm ball for $p \geq 1$ is $\B_p = \{z \mid \|z\|_p \leq 1\}$.
A \textit{polytope} $\P \subset \R^n := \{z \in \R^n \mid Fz \leq g\}$ is the (bounded) intersection of a finite number of halfspaces. \revision{Scaling of polytopes by a factor $\lambda > 0$ is defined as $\lambda \P = \{\lambda z \in \R^n \mid Fz \leq g\} = \{z \in \R^n \mid Fz \leq \lambda g\}$.} Given a matrix $F$ and a vector $g$, the $i$th row of $F$ is denoted $F^{(i)T}$ and the $i$th component of $g$ is $g^{(i)}$. The interior of any set $\Q$ is denoted $\int \Q$. The value of a variable $y$ at a time interval $t$ is denoted $y_t$. \newrevision{A state or control trajectory of length $\tau$ is written as the vector $\textbf{x} = \m{x_1^T,\ldots,x_{\tau}^T}^T \in \R^{n\tau}$ or $\textbf{u} = \m{u_0^T,\ldots,u_{\tau-1}^T}^T \in \R^{m\tau}$.} The column vector of all ones is \textbf{1}.
\newrevision{The symbol $\circ$ denotes function composition.}
\section{Problem Formulation} \label{sec:pf}


In this paper, we consider the problem of regulating discrete-time dynamical systems of the form \begin{align}
    x_{t+1} = Ax_t + Bu_t + d_t \label{eqn:2-26-5}
\end{align} where $x_t \in \R^n$ is the system state at time $t$, $u_t \in \R^m$ is the control input, and $d_t \in \R^n$ is an uncertain input that captures exogenous disturbances and/or linearization error (if the true system dynamics are nonlinear) \cite{Boyd1994}. We assume the pair $(A,B)$ is stabilizable.
\revision{The input constraints (actuation limits) are $\U = \{u \in \R^m \mid F_u u \leq g_u\}$ while the state constraints arising from safety-critical engineering considerations are $\X = \{x \in \R^n \mid F_x x \leq g_x\}$.}


We consider the problem of operating the system \eqref{eqn:2-26-5} using finite-horizon model predictive control. The goal is to choose, given initial condition $x_0 \in \X$, a sequence of inputs $\textbf{u}$ of length $\tau$ that minimizes the cost of operating the system while respecting the operational constraints.


However, since the disturbances $d_t$ are unknown ahead of time, the designer must carefully consider how to achieve both optimality and constraint satisfaction. 
Robust MPC literature contains many ways to handle the presence of disturbances in both the cost and constraints \cite{Bemporad1999}. For example, the \textit{certainty-equivalent}  approach \cite{Alessio2009} considers only the nominal system trajectory, while the \textit{min-max} approach \cite{Grancharova2009} considers the worst-case disturbance. Interpolating between these two extremes, the \textit{tube-based} approach \cite{Langson2004} considers the cost of a nominal trajectory while guaranteeing that the true trajectory satisfies constraints. A \textit{stochastic} point of view in \cite{Farina2016} considers the disturbance as a random variable and minimizes the expected cost while providing probabilistic guarantees for constraint satisfaction. 

\revision{In most robust MPC formulations, the set of possible disturbances is modeled as either a finite set, a bounded set, or a probability distribution \cite{Saltk2018}.} In this paper, we assume the disturbances lie in a closed and bounded set $\D := \{d \in \R^n \mid F_d d \leq g_d\}$. \revision{In order to ensure constraint satisfaction, we operate the system within a \textit{robust control invariant set} (RCI) $\S \subseteq \X$, defined as a set of initial conditions for which there exists a feedback policy in $\U$ keeping all system trajectories in $\S$, under any disturbance sequence in $\D$ \cite{Blanchini2015}. In our simulations, we used approximately-maximal RCIs computed with the semidefinite program from \cite{Liu2015}.} 

\revision{With $\S := \{x \in \R^n \mid F_s x \leq g_s\}$, we define the \textit{target set} $\T$ as $\{x \in \R^n \mid x + d \in \S, \ \forall\ d \in \D\} = \{x \in \R^n \mid F_s x \leq \tilde{g}_s\}$ where for each row $i$, $\tilde{g}_s^{(i)} = g_s^{(i)} - \max_{d \in \D} F_s^{(i)T}d$ \cite{Blanchini2015}}. \newrevision{Any policy that maps $\S$ to $\T$ under the nominal dynamics will map $\S$ to itself under the true dynamics, rendering $\S$ robustly invariant.} \revision{By constraining the nominal state to the target set, robust constraint satisfaction is guaranteed for the first time step. Since $\S$ is RCI, this is sufficient for keeping closed-loop trajectories inside $\S$. Under this formulation, the MPC problem is posed as follows, given initial state $x_0$:
\begin{subequations} \label{eqn:2-26-7} \begin{gather}
    \min_{\u} \sum_{k=0}^{\tau-1} l(x_k,u_k) + l_F(x_{\tau}) \label{eqn:2-26-8}\\
    \text{subject to $\forall\ k$: } x_{k+1} = Ax_k + Bu_k \label{eqn:2-26-9}\\
    x_{k+1} \in \T \label{eqn:2-26-10}\\
    u_k \in \U \label{eqn:2-27-9}
\end{gather} \end{subequations} where $l$ and $l_F$ are stage and terminal costs that are differentiable but possibly nonlinear or even non-convex.} \newrevision{Although \eqref{eqn:2-26-7} differs from the standard tube-based approach, the techniques introduced in this paper can be applied to a variety of MPC formulations.}


\revision{In this paper, we seek to derive a safe feedback policy $\pi_\theta: \R^n \rightarrow \R^m$ that approximates the explicit solution to \eqref{eqn:2-26-7} by first approximating the optimal control sequence with a function $\mu_\theta: \R^n \rightarrow \R^{m\tau}$ and then implementing the first action of the sequence in the closed loop. In practice, any MPC policy implemented in closed loop must be stabilizing and recursively feasible. Recursive feasibility is the property that closed-loop trajectories generated by the MPC controller will not lead to states in which the MPC problem is infeasible. This property is guaranteed when $\S$ is RCI \cite{Blanchini2015}. If recursive feasibility is not guaranteed, then a backup controller must be developed or a control sequence that is feasible for the most immediate time steps can be used. There is suggestion in the literature that the latter approach performs quite well in practice~\cite{Wang2010}, but the theoretical aspects remain open. In terms of stability, recursive feasibility guarantees that trajectories will remain within a bounded set. Since this work focuses on constraint satisfaction, we do not consider stricter notions of stability.}

\section{Phase I: Finding a Feasible Point} \label{sec:p1}

\revision{The feasible set of \eqref{eqn:2-26-7} is a polytope $\F(x_0) \subseteq \R^{m\tau}$, defined by the following inequalities in $\u$: \begin{subequations} \begin{align}
        H_s(M_0 x_0 + M_u \textbf{u}) &\leq \tilde{h}_s, \label{eqn:5-22-22-4}\\
        H_u \textbf{u} &\leq h_u \label{eqn:5-22-22-5}
    \end{align} \end{subequations} where $H_s,H_u,M_0,M_u,\tilde{h}_s,$ and $h_u$ are block matrices and vectors derived from the system dynamics and constraints}.
\revision{In this paper, we assume that $\F(x_0)$ \newrevision{has nonempty interior} for all $x_0 \in \S$. Since the state constraints $\S$ form an RCI, $\F(x_0)$ is already guaranteed to be nonempty, and the assumption of nonempty interior is only marginally more restrictive.} 

\revision{The gauge map technique introduced in \cite{Tabas2021a} provides a way to constrain the outputs of a neural network $\mu_\theta: \R^n \rightarrow \R^{m\tau}$ to $\F(x_0)$ without a projection or penalty function, but $\F(x_0)$ must contain the origin in its interior. If this is not the case, then we must temporarily ``shift'' $\F(x_0)$ by subtracting any one of its interior points. In this section, we discuss several ways to reduce the complexity of finding an interior point.}

\revision{We begin by considering the feasibility problem for the one-step safe action set defined as $\V(x_0) = \{u \in \R^m \mid u \in \U, Ax_0 + Bu \in \T\},$ which is guaranteed to have an interior point by the assumption on $\F(x_0).$ \newrevision{One way to find an interior point of $\V(x_0)$ is to minimize the maximum constraint violation:}} \begin{subequations} \label{eqn:3-17-3} \begin{gather}
    \min_{u,s} s \label{eqn:8-23-22-1}\\
    \text{subject to: } F_s(Ax_0 + Bu) \leq \tilde{g}_s + s\textbf{1} \label{eqn:5-24-22-1}\\
    F_u u \leq g_u + s\textbf{1} \label{eqn:5-24-22-2}
\end{gather} \end{subequations} \revision{which has an optimal cost $s^* \leq 0$ if \newrevision{$\V(x_0)$ is nonempty}, and $s^* < 0$ if \newrevision{$\V(x_0)$ has nonempty interior} \cite{Boyd2009}. To avoid solving a linear program online during closed-loop implementation, the solution to \eqref{eqn:3-17-3} can be stored as a piecewise affine (PWA) function $\pi_0(x_0):\R^n \rightarrow \R^m$ \cite{Jones2007}. Although solutions to multiparametric LPs can be demanding on computer memory, we take advantage of the fact that feasibility problems have low accuracy requirements: any \newrevision{suboptimal} solution to \eqref{eqn:3-17-3} that achieves a cost $s < 0$ for all $x_0 \in \S$ is acceptable.
}
\newrevision{ \begin{definition}
A function $\pi_0: \R^n \rightarrow \R^m$ is said to solve \eqref{eqn:3-17-3} if,  for all $x_0 \in \S$, the optimal cost of \eqref{eqn:3-17-3} is negative when the decision variable $u$ is fixed at $\pi_0(x_0)$. 
\end{definition}}

\revision{Existing techniques for approximate multiparametric linear programming \cite{Filippi2004}, especially those that generate continuous solutions \cite{Spjotvold2005}, can be used to reduce the memory requirements of offline solutions to \eqref{eqn:3-17-3}.} 

\ifx
\newrevision{To characterize the set of allowable approximations to the optimal solution of \eqref{eqn:3-17-3}, we pose the following feasibility problem: \begin{align}
    \min_{u,s} 0
    \text{ s.t. } s < 0, \eqref{eqn:5-24-22-1}, \eqref{eqn:5-24-22-2}. \label{eqn:8-24-22-1}
\end{align} We will say that a function $\pi_0$ solves \eqref{eqn:8-24-22-1} if $\pi_0(x_0)$ solves \eqref{eqn:8-24-22-1} for all $x_0 \in \S$.}
\fi

\revision{To show just how far one can go with reducing complexity, we will construct an affine (rather than PWA)} 
\newrevision{function that solves \eqref{eqn:3-17-3},} 
\revision{for the system studied in Section \ref{sec:sims}.} 
Let $\pi_0(x_0) = Wx_0 + w$. If $W \in \R^{m \times n}$ and $w \in \R^m$ satisfy
\begin{subequations} \label{eqn:3-17-1}
\begin{align}
    F_x(Ax_0 + B(Wx_0+w)) &< \tilde{g}_x \\
    F_u (Wx_0+w) &< g_u 
\end{align} \end{subequations} for all $x_0 \in \S$, then $\pi_0(x_0) = Wx_0 + w$ \newrevision{solves \eqref{eqn:3-17-3}}. 
\newrevision{The following optimization problem} can be solved to find $W$ and $w$ or certify that none exists. Let $\Y(s) = \{x_0 \in \R^n \mid F_s(Ax_0 + B(Wx_0+w)) \leq \tilde{g}_s + s \textbf{1}, F_u (Wx_0 + w) \leq g_u + s \textbf{1}\}.$ \newrevision{If the optimal cost of} \begin{gather}
    \min_{W,w,s} s 
    \text{ subject to } \newrevision{\S} \subseteq \Y(s) \label{eqn:3-17-4}
\end{gather} \newrevision{is negative, then 
\eqref{eqn:3-17-1} holds for all $x_0 \in \S$, thus $\pi_0$ solves \eqref{eqn:3-17-3}. This happens to be the case for the example in Section \ref{sec:sims}, taken from \cite{Zeilinger2011}}. The constraint in \eqref{eqn:3-17-4} is a polytope containment constraint in halfspace representation, \newrevision{thus \eqref{eqn:3-17-4} can be solved as} a linear program \cite{Sadraddini2019}.

\revision{Now consider the feasibility problem for $\F(x_0)$, which is obtained by replacing \eqref{eqn:5-24-22-1} and \eqref{eqn:5-24-22-2} with \eqref{eqn:5-22-22-4} and \eqref{eqn:5-22-22-5}, and changing the optimization variable from $u \in \R^m$ to $\u \in \R^{m\tau}$. One would naturally expect the complexity of the PWA solution to this feasibility problem to increase rapidly with the time horizon $\tau$, as more decision variables and constraints are added. However, the next proposition shows that the cardinality of the stored partition can be made constant in $\tau$.}
\revision{\begin{proposition}[Phase I solution] \label{prop:2-27-1} If $\pi_0$ \newrevision{solves \eqref{eqn:3-17-3},} then the \newrevision{vector $\mu_0(x_0) := \m{\pi_0(x_0)^T,\ldots,\pi_0(x_{\tau-1})^T}^T$}, where $x_{k+1} = Ax_k + B \pi_0(x_k)$, is an interior point of $\F(x_0)$ for any $x_0 \in$ \newrevision{$\S$}.
\end{proposition}
\begin{proof} \newrevision{If $\pi_0$ solves \eqref{eqn:3-17-3}, then $\pi_0(x) \in \int \V(x)$ for all $x \in \S$. Applying the definition of $\V$ in an inductive argument,} it is straightforward to show that the \newrevision{state trajectory associated with $\mu_0(x_0)$ is entirely contained } 
in $\S$. Fix any such trajectory $\{x_1,\ldots,x_\tau\} \subset \S$ originating from $x_0 \in \S$ under policy $\pi_0$. For any $k \in \{1,\ldots,\tau\}$, \newrevision{$x_k \in \S$ implies $\pi_0(x_k) \in \int \V(x_k)$, which} implies $\pi_0(x_k) \in \int \U$ and $Ax_k + B\pi_0(x_k) \in \int \T$. Since this holds for all $k$, the constraints defining $\F(x_0)$ hold strictly at $\mu_0(x_0)$. 
\end{proof}}





In our simulations on the example from \cite{Zeilinger2011},  $\eqref{eqn:3-17-4}$ was feasible with negative optimal cost, \newrevision{meaning that a polyhedral partition of the state space was not needed}
(see Section \ref{sec:sims}). This indicates that the minimum number of regions \newrevision{in a polyhedral state space partition associated with a PWA
solution to \eqref{eqn:3-17-3}} is in general very small relative to the number of regions in an explicit solution to \eqref{eqn:2-26-7}. 

\section{Phase II: Optimizing Performance} \label{sec:p2}

In this section, we construct a class of policies from $x_0 \in \S$ to $\F(x_0)$, that can be trained using standard machine learning packages. Although it is difficult to constrain the output of a neural network to an arbitrary polytope such as $\F(x_0)$, it is easy to constrain the output to the hypercube $\B_\infty$ by applying a clamping function elementwise in the output layer. We apply a mapping between polytopes that is closed-form, differentiable, and bijective. This mapping establishes an equivalence between $\B_\infty$ and $\F(x_0)$, allowing one to constrain the outputs of the policy to $\F(x_0)$. The mapping from $\B_\infty$ to $\F(x_0)$ is called the \textit{gauge map}. The concept is illustrated in Figure \ref{fig:nn_diagram}. 

\begin{figure}
    \centering
    \includegraphics[width=7cm]{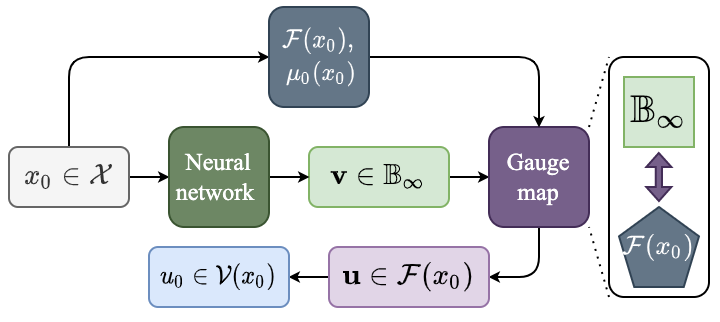}
    \caption{The proposed control policy uses a neural network combined with the Phase I solution and a \textit{gauge map} to constrain the decision $\u$ to the MPC feasible set $\F(x_0)$. The first action from the sequence $\u$ is extracted and implemented. On the right, the action of the gauge map is illustrated.}
    \label{fig:nn_diagram}
\end{figure}

We begin constructing the gauge map by introducing some preliminary concepts. A \textit{C-set} is a convex, compact set that contains the origin as an interior point. The \newrevision{\textit{gauge function} with respect to C-set $\P \subset \R^n$, denoted $\gamma_{\P}: \R^n \rightarrow \R_+$, is the function whose sublevel sets are scaled versions of $\P$. Specifically,} the gauge of a vector $v$ with respect to $\P$ is given by $\gamma_\P(v) = \inf\{\lambda \geq 0 \mid v \in \lambda \P\}.$ If $\P$ is a polytopic C-set given by $\{v \in \R^k \mid Fv \leq g\}$, then \newrevision{$\gamma_{\P}$ is the pointwise maximum over a finite set of affine functions \cite{Tabas2021a}:} \begin{align} \gamma_\P(v) = \max_i \frac{F^{(i)T}v}{g^{(i)}}. \label{eqn:5-25-22-2} \end{align} 
Given two C-sets $\P$ and $\Q$, the \textit{gauge map} $G: \P \rightarrow \Q$ is 
\begin{align} G(v \mid \P,\Q) = \frac{\gamma_\P(v)}{\gamma_\Q(v)} \cdot v. \label{eqn:3-21-1} \end{align} 
\newrevision{This function maps level sets of $\gamma_{\P}$ to level sets of $\gamma_{\Q}$.}

\begin{proposition} \label{prop:2-27-2}
Given two polytopic C-sets $\P$ and $\Q$, the gauge map $G: \P \rightarrow \Q$ is subdifferentiable and bijective. Further, given a function $\pi_0$ from Proposition \ref{prop:2-27-1}, the set $\tilde{\F}(x) := [\F(x)-\pi_0(x)]$ is a C-set for all $x \in \S$.
\end{proposition}

\begin{proof}
The properties of subdifferentiability and bijectivity come from \cite{Tabas2021a}. For the C-set property, fix $x \in \S$. Since $\S,$ $\U,$ and $\D$ are convex and compact, so is $\F(x)$. Since $\mu_0(x)$ is an interior point of $\F(x),$ the set $\tilde{\F}(x)$ contains the origin as an interior point and is therefore a C-set.
\end{proof}

We now use the gauge map in conjunction with the Phase I solution to construct a neural network whose output is confined to $\F(x_0).$
Let $\psi_\theta: \S \rightarrow \B_\infty$ be a neural network parameterized by $\theta$. \revision{A safe policy is constructed by composing the gauge map $G: \B_\infty \rightarrow \tilde{\F}(x_0)$ with $\psi_\theta$, then adding $\mu_0(x_0)$ to map the solution into $\F(x_0)$:}
\begin{align}
    \mu_\theta(x_0) = G(\cdot \mid \B_\infty, \tilde{\F}(x_0)) \circ \psi_\theta(x_0)  + \mu_0(x_0). \label{eqn:3-24-1}
\end{align} \revision{Computing the gauge map online simply requires evaluating $H_s M_0 x_0$ from \eqref{eqn:5-22-22-4} as well as the operations in \eqref{eqn:5-25-22-2}.}

The function $\mu_\theta$ has several important properties for approximating the optimal solution to \eqref{eqn:2-26-7}. First, it leverages the universal function approximation properties of neural networks \cite{Hornik1989} along with the bijectivity of the gauge map (Proposition \ref{prop:2-27-2}) to explore all interior points of $\F(x_0).$ This is an advantage over projection-based methods \cite{Chen2018d} which may be biased towards the boundary of $\F(x_0)$ when the optimal solution may lie on the interior. Second, $\mu_\theta$ is evaluated in closed form, and its outputs are constrained to $\F(x_0)$ without the use of an optimization layer \cite{Maddalena2020} that may have high computational overhead. Finally, the subdifferentiability of the gauge map (Proposition \ref{prop:2-27-2}) enables selection of parameter $\theta$ using standard automatic differentiation techniques.


\subsection*{Optimizing the parameter $\theta$}

Similar to the approach taken in \cite{Akesson2006}, we optimize $\theta$ by sampling $x \in \S$ and applying stochastic gradient descent. At each iteration, a new batch of initial conditions $\{x_0^j\}_{j=1}^M$ is sampled from $\S$ and the loss is computed as \begin{align}
    J(\theta) = \frac{1}{M} \sum_{j=1}^M \sum_{k=0}^{\tau-1} l(x_k^j,u_k^j) + l_F(x_{\tau}^j) \label{eqn:2-27-10}
\end{align} with the control sequences $\u^j$ given by $\mu_\theta(x_0^j)$ and state trajectories $\x^j$ generated according to the nominal dynamics. \revision{The parameters $\theta$ are updated in the direction of $\nabla_\theta J$, which is easily computed using automatic differentiation~\cite{Gune2018}.}

\section{Simulations} \label{sec:sims}

\subsection{Test systems}



We simulate the proposed policy using a modified example from \cite{Zeilinger2011} with $n=3$, $m=2,$ and $\tau=5$. The system matrices, constraints, costs, and Phase I solution (found using~\eqref{eqn:3-17-4}) are given below:
\begin{gather}
    A = \m{-.5&.3&-1 \\.2&-.5&.6\\1&.6&-.6},\ B = \m{-.601&-.890\\.955&-.715\\.246&-.184}, \\
    \|x\|_\infty \leq 5,\ \|u\|_\infty \leq 1,\ \|d\|_\infty \leq 0.1 \label{eqn:2-27-8}, \\
    l(x,u) = \|x\|_2^2 + c_1\|u\|_2^2, \ l_F(x) = c_2\|x\|_2^2 \\
    W = \m{0.116 &  0.210 & -0.370 \\
       -0.320 & -0.104 & -0.122}, w = \m{-0.157 \\ -0.0533} \nonumber
\end{gather} where $c_1$ and $c_2$ are positive constants. Although quadratic costs are used in the simulations, the proposed method can work with any differentiable cost. 

\revision{We evaluate the performance of a given policy in both open- and closed-loop experiments. In the open-loop experiments, we evaluate the MPC cost \eqref{eqn:2-26-8} and compare it to the optimal cost. The fraction suboptimality is \begin{align}
    \delta = \frac{c_{nn} - c_{mpc}}{c_{mpc}} \label{eqn:5-25-22-1}
\end{align} where $c_{nn}$ is the average cost \eqref{eqn:2-26-8} incurred by the control sequence $\mu_\theta$ on a validation set $\{x_0^j\}_{j=1}^{N_{val}} \subset \S$ and $c_{mpc}$ is the optimal cost.}

\revision{In the closed-loop experiments, we evaluate the performance of a policy $\pi_\theta(x_t): \R^n \rightarrow \R^m, t \geq 0$ which is derived from $\mu_\theta(x_t)$ by taking the first action in the sequence. We simulate \eqref{eqn:2-26-5} for $T \gg \tau$ time steps.}
The trajectory cost in the closed-loop experiments is computed as 
    $\sum_{t = 0}^{T-1} l(x_t,u_t) + l_F(x_T)$ 
and the disturbance is modeled as an autoregressive sequence \cite{Srinath1995}, 
$d_{t+1} = \alpha d_t + (1-\alpha) \hat{d}$ 
where $\alpha \in (0,1)$ and $\hat{d}$ is drawn uniformly over $\D$.

\subsection{Benchmarks}

We compare the proposed method to two of the most common approaches for learning a solution to \eqref{eqn:2-26-7}. The first benchmark is a penalty-based approach \cite{Drgona2020} which enforces the constraints \eqref{eqn:2-26-10} and \eqref{eqn:2-27-9} by augmenting the cost \eqref{eqn:2-27-10} with a linear penalty term on constraint violations given by $\beta \cdot \max\{0,F_x x_t-\tilde{g}_x\}$ where the $\max$ is evaluated elementwise and $\beta > 0$. Since the penalty-based approach does not encode state constraints in the policy, the policy is constrained to the Cartesian product $\U^\tau = \prod_{k=0}^{\tau-1} \U$ using scaled $\tanh$ functions elementwise.

The second benchmark is a projection-based approach \cite{Chen2018d} which constrains the policy to the set $\F(x_0)$ by solving a convex quadratic program in the output layer of a neural network \cite{Agrawal2019}. 
The optimization layer $\textbf{v} \rightarrow \textbf{u}$ returns
\begin{gather*}
    \underset{\u}{\arg \min} \|\textbf{v} - \u\|_2^2
    \text{ subject to } \u \in \F(x_0).
\end{gather*}

Another class of approaches to learning-based MPC seeks to learn the optimal solution to \eqref{eqn:2-26-7} using regression \cite{Parisini1995, Domahidi2011,Maddalena2020}. Specifically, data-label pairs $(x_0,u_0^*)$ are generated by sampling $x_0$ from $\S$, solving \eqref{eqn:2-26-7} for each sample, and extracting $u_0^*$ from the optimal solution $\u^*$. Then, a neural network or other function approximator is trained to learn the relationship between $x_0$ and $u_0^*$. Performance and constraint satisfaction are handled e.g. by bounding the approximation error with respect to the MPC oracle. We do not compare against this type of approach since it requires a large number of trained samples, making it difficult to compare with our and the other unsupervised examples. 

\subsection{Neural network design}

\revision{The neural networks were designed with $n$ inputs, $m\tau$ outputs, and two hidden layers with rectified linear unit (ReLU) activation functions. The width of the networks was chosen during hyperparameter tuning. In particular, we performed 30 iterations of random search over the width of the network (number of neurons per hidden layer) $\in \{64,\ldots,1024\}$, the batch size (number of initial conditions, $M$) $\in \{100,\ldots,3000\}$ and the learning rate (LR, the step size for gradient descent) $\in [10^{-5},10^{-3}]$. For each set of hyperparameters under consideration, we computed the validation score using \eqref{eqn:5-25-22-1} with $N_{val}=100$. The hyperparameters after tuning are reported in Table \ref{table:hparams}.}

\begin{table}
\centering
\caption{Hyperparameters for the three neural networks.}
\begin{tabular}{ |c|c|c|c| } 
 \hline
 Type & Width & LR & $M$\\ 
 \hline
 Gauge & $859$ & $4.7 \times 10^{-4}$ & $1655$ \\
 Penalty & $318$ & $8.7 \times 10^{-4}$ & 133 \\
 Projection & $956$ & $9.0 \times 10^{-5}$ & $813$ \\
 \hline
\end{tabular}
\label{table:hparams}
\end{table}

\subsection{Simulation results}
Here we compare our proposed approach (Gauge NN), the penalty-based approach (Penalty NN), the projection-based approach (Projection NN) and the ``ground truth'' obtained by solving  \eqref{eqn:2-26-7} online in \texttt{cvxpy}. \revision{The results of the open-loop experiments are shown in Table \ref{table:open_loop}, with performance computed relative to the optimal MPC solution using \eqref{eqn:5-25-22-1} with $N_{val}= 100$ trials. The proposed Gauge NN achieves lower cost compared to the projection-based method, and has a much lower computational complexity (solve time is only 6\% of projection).  Table \ref{table:open_loop} only compares the NNs with safety guarantees because constraint violations are not accounted for in \eqref{eqn:5-25-22-1}.} 

\begin{table}
\centering
\caption{Open-loop test results.}
\begin{tabular}{ |c|c|c| } 
 \hline
 Type & $\delta$ \eqref{eqn:5-25-22-1} & Solve time (sec)\\ 
 \hline
 Gauge & 0.007 & .0015\\ 
 Projection & 0.010 & .024\\
 \hline
\end{tabular}
\label{table:open_loop}
\end{table}

Figure \ref{fig:train} shows the training curves for each type of network. The lower training cost achieved by the Gauge NN illustrates that it can be more efficient to explore the interior of the feasible set than the boundary. \revision{Since the MPC cost in the simulations is strictly convex, solutions with lower cost are closer to the optimal solution.}


\begin{figure}
    \centering
    \includegraphics[width=8cm]{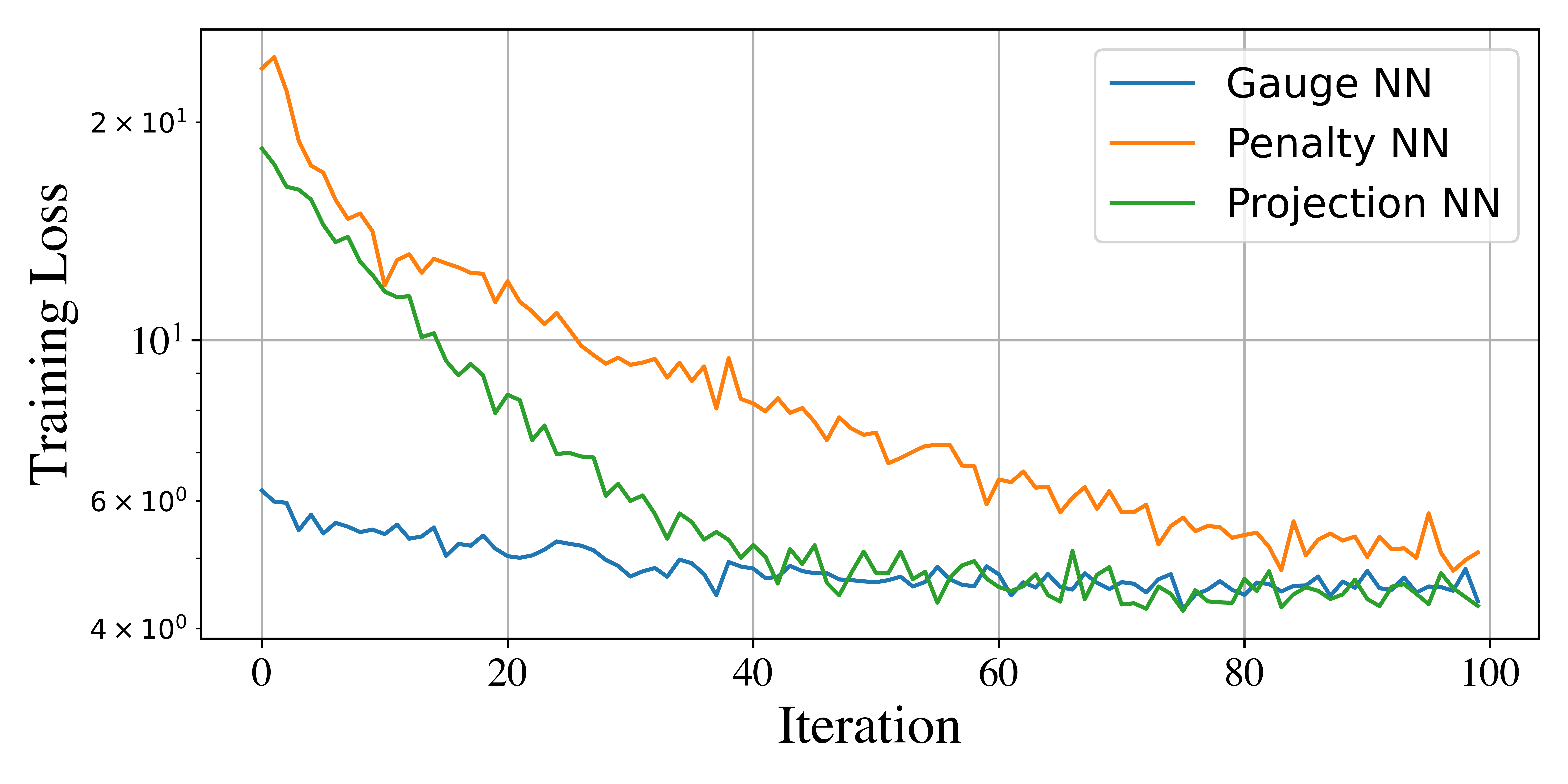}
    \caption{Training trajectories for the three types of neural netwokrs. Our proposed Gauge-based approach achieves lower cost at a much faster rate.}
    \label{fig:train}
\end{figure}

Figure \ref{fig:pareto_3} compares the policies in terms of computation time and test performance. The box-and-whisker plots indicate the range of performance over 100 test trajectories of length $T=50$, while the vertical position of each box indicates the average time to compute a control action. Of the policies with safety guarantees (Gauge NN, Projection NN, and online MPC), \revision{the Gauge NN achieves Pareto efficiency in terms of average solve time and median trajectory cost.} 
Our intuition behind the high performance of the neural networks is that \eqref{eqn:2-26-7} is a heuristic and the unsupervised learning approach can lead to better closed-loop policies.

\begin{figure}
    \centering
    \includegraphics[width=8cm]{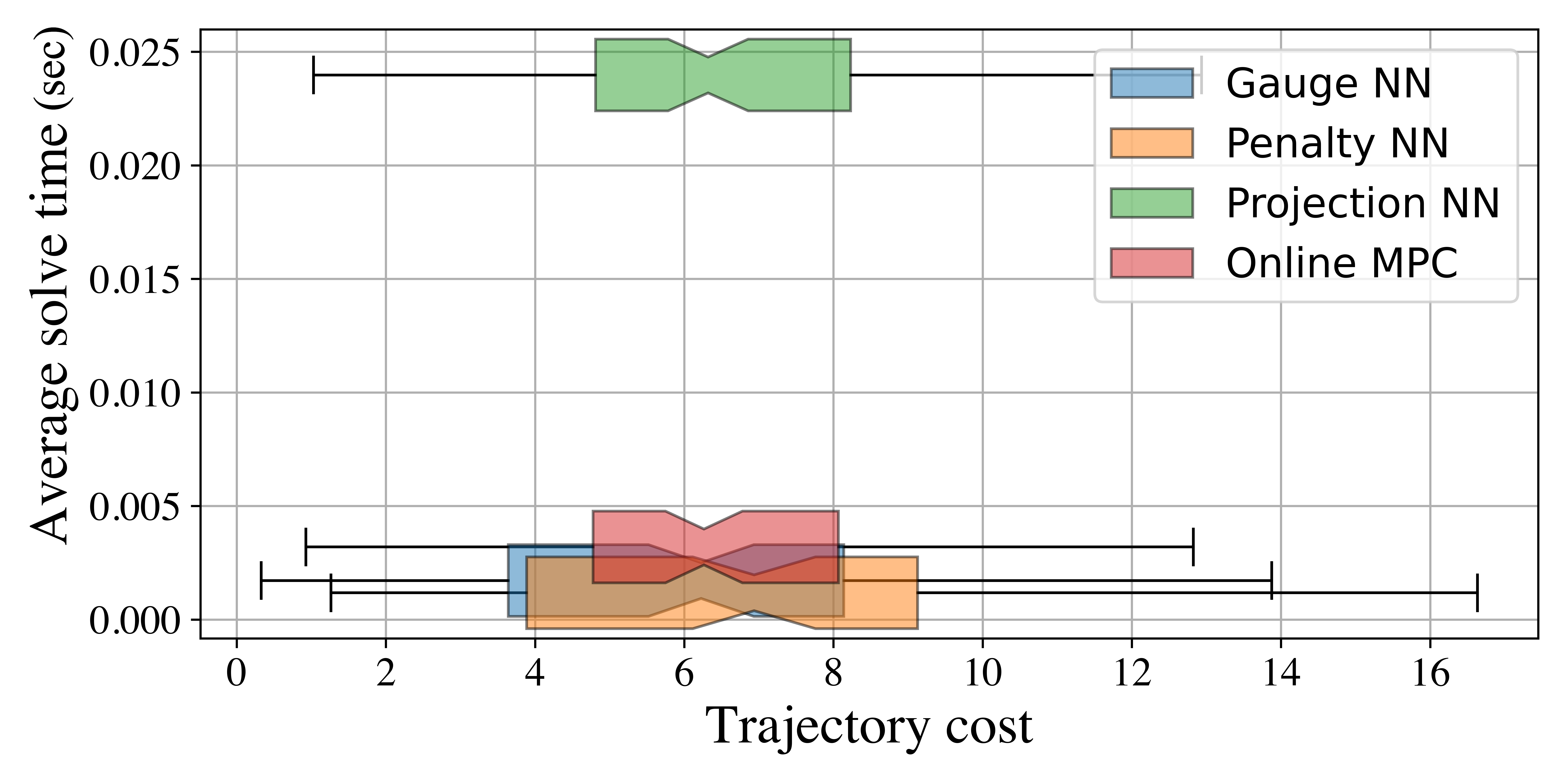}
    \caption{Solve time vs. trajectory cost for the networks under consideration applied to the 3-state system. The Gauge NN is Pareto-efficient in terms of cost and computation time compared to the other techniques with safety guarantees (Online MPC and Projection NN).}
    \label{fig:pareto_3}
\end{figure}


\ifx

\subsection{Extension to scenario-based MPC}
A drawback to both standard online MPC and neural network based methods is that they minimize the cost of a nominal trajectory, where the disturbance is not explicitly taken into account. 
One way to address the issue of disturbances is via the \emph{scenario approach} \cite{Lucia2018}, in which the cost \eqref{eqn:2-26-8} is averaged over $N$ sequences of disturbances sampled from $\D$ according to \eqref{eqn:3-19-1}. The MPC \eqref{eqn:2-26-7} is replaced by \begin{subequations}
\begin{gather}
    \min_{\hat{\u}_t} \frac{1}{N} \sum_{i=1}^N \sum_{k=t}^{t+\tau-1} l(\hat{x}_k^i,\hat{u}_k) + l_F(\hat{x}_{t+\tau}^i) \\
    \text{s.t. }
    \hat{x}_{k+1}^i = A\hat{x}_k^i + B\hat{u}_k + d_k^i \ \forall\ i,k;\ \hat{x}_t^i = x_t  \label{eqn:3-20-1}\\
    \tilde{x}_{k+1} = A\tilde{x}_k + B\hat{u}_k \ \forall\ k;\ \tilde{x}_t = x_t  \label{eqn:3-21-2}\\
    \tilde{x}_{k+1} \in \T,\ \hat{u}_k \in \U \ \forall\ i,k. \label{eqn:3-21-3}
\end{gather} \end{subequations} For feasibility under any disturbance sequence, \eqref{eqn:3-21-3} constrains the nominal trajectory in \eqref{eqn:3-21-2} to the target set.
For the NN, the policy class \eqref{eqn:3-24-1} remains unchanged but the loss function \eqref{eqn:2-27-10} is averaged over scenarios, becoming \begin{align}
    J(\theta) = \frac{1}{MN} \sum_{i=1}^N \sum_{j=1}^M \sum_{k=0}^{\tau-1} l(\hat{x}_k^{i,j},\hat{u}_k^{j}) + l_F(\hat{x}_{\tau}^{i,j})
\end{align} with the controls $\hat{\u}_0^{j}$ given by $\pi_\theta(x_0^j)$ and trajectories by \eqref{eqn:3-20-1} for each scenario $i$ and each initial condition $j$.

Figure \ref{fig:pareto_s} shows that when noise is considered using the scenario approach, the MPC computation times increase sharply while the Gauge NN remains competitive in terms of median trajectory cost. The efficiency gained by taking an interior point perspective is multiplied when the designer wishes to account for disturbances using a scenario approach.

\begin{figure}
    \centering
    \includegraphics[width=8cm]{Figures/pareto_3n2m_True.png}
    \caption{Solve time vs. trajectory cost in the scenario-based setting, for the 3-state system. The computation time for online MPC increases dramatically, while Gauge NN achieves both lowest computation time and trajectory cost.}
    \label{fig:pareto_s}
\end{figure}

\fi
\section{Conclusion}

In this paper, we provided an efficient way of exploring the interior of the MPC feasible set for learning-based approximate explicit MPC, and demonstrated the performance and computational gains that can be achieved by approaching the problem from the interior. The paradigm relies on a Phase I solution that exploits the structure of the MPC problem and a Phase II solution that features a projection-free feasibility guarantee. The results compare favorably against common approaches that use unsupervised learning, as well as against the oracle itself used in supervised approaches. Future work includes applications to MPC problems with convex but non-polytopic constraint sets, and to distributed settings. 

\bibliography{references2}
\bibliographystyle{ieeetr}

\end{document}